\documentclass[10pt, final, oneside, twocolumn, letter]{IEEEtran}
\usepackage{graphicx}
\usepackage{amsmath, amsthm, mathrsfs}
\usepackage{hyperref}
\usepackage{amssymb}  
\usepackage{color}

\newcommand{\tr}{\mathbf{tr}}
\newcommand{\E}{\mathbf{E}}
\newtheorem{theorem}{\textbf{Theorem}}
\newtheorem*{theorem*}{\textbf{Theorem}}

\newtheorem{assumption}{\textbf{Assumption}}

\title{Nonlinear Estimation for Position-Aided Inertial Navigation Systems}
\author{Soulaimane Berkane and Abdelhamid Tayebi
\thanks{This work was supported by the National Sciences and Engineering Research Council of Canada (NSERC), under the grants NSERC-DG RGPIN-2020-04759 and NSERC-DG RGPIN-2020-0627. S. Berkane (\href{mailto:soulaimane.berkane@uqo.ca}{soulaimane.berkane@uqo.ca}) is with the Department of Computer Science and Engineering, University   of   Quebec   in   Outaouais, Gatineau, QC, Canada. A. Tayebi (\href{mailto:atayebi@lakeheadu.ca}{atayebi@lakeheadu.ca}) is with Department of Electrical Engineering, Lakehead University, Thunder Bay, ON, Canada.}}   
\begin{document}
\maketitle
\begin{abstract}
In this work we solve the position-aided 3D navigation problem using a nonlinear estimation scheme. More precisely, we propose a nonlinear observer to estimate the full state of the vehicle (position, velocity, orientation and gyro bias) from IMU and position measurements. The proposed observer does not introduce additional auxiliary states and is shown to guarantee semi-global exponential stability without any assumption on the acceleration of the vehicle. The performance of the observer is shown, through simulation, to overcome the state-of-the-art approach that assumes negligible accelerations. 
\end{abstract}
\section{Introduction}
Inertial navigation systems (INS) are essential devices that allow the localization and control of autonomous vehicles and robot platforms \cite{Titterton2004StrapdownTechnology}. Classical INS fuse measurements from on-board accelerometers and gyroscopes (typically included in an Inertial Measurement Unit (IMU)), to continuously compute the position, velocity and orientation of a vehicle without any external reference measurement. However, this classical approach (known also as dead reckoning) suffers from drift due to measurement errors and unknown initial conditions \cite{Woodman2007AnNavigation}. As a result, inertial navigation systems are often assisted by position sensors such as the Global Positioning System (GPS), which allow for the correction of position estimates over time, resulting in small and bounded estimation errors, see, {\it e.g.,} \cite{vik2001nonlinear,Farrell2008AidedSensors,Grip2013}. Other type of sensors that can provide range (distance)  measurements  to  known  source  points can also be used to provide position information such as GPS pseudo-ranges \cite{johansen2015nonlinear,Johansen2017NonlinearMeasurements,Bryne2017NonlinearAspects} or Ultra Wide-Band (UWB) radio technology \cite{Gryte2017RobustGPS,Hamer2018Self-calibratingLocalization}. Kalman-type filters, such as \cite{sabatini2006quaternion,Farrell2008AidedSensors,Crassidis2006Sigma-pointNavigation,Whittaker2017InertialRepresentations}, are considered industry-standard solutions for inertial navigation systems. However, these filters are often based on linearization assumptions and may fail when the initial estimation errors are large.  On the other hand, nonlinear  observers have been developed for autonomous navigation applications; see for instance \cite{Johansen2017NonlinearMeasurements,hansen2018nonlinear,Bryne2017NonlinearAspects,berkane2019position,WangTayebiTAC2020,WangTayebiAuto2020,berkane2021nonlinear}. The advantage of the nonlinear observers is their theoretically proven stability guarantees, as well as their computational simplicity compared to the stochastic filters. 

In this work, we propose a nonlinear observer for the simultaneous estimation of the position, linear velocity, attitude, and gyro bias of a rigid body system. The estimator relies on IMU and inertial position measurements which are used in typical navigation scenarios. The attitude estimates are directly obtained on the Special Orthogonal group of rotations $\mathbb{SO}(3)$, thus avoiding any singularities or ambiguities related to the use of other attitude parameterizations. The proposed observer, evolving on $\mathbb{SO}(3)\times\mathbb{R}^9$ is shown to guarantee semi-global exponential stability, which is the strongest stability result  that can be achieved on this state space using smooth observers. In contrast to \cite{Grip2013,Johansen2017NonlinearMeasurements,berkane2019position,berkane2021nonlinear}, the proposed observer does not require the introduction of an auxiliary $3$-dimensional state in the estimation scheme, thus reducing the computational burden associated with the real-time implementation of the estimator.  The paper is structured as follows. Preliminaries are provided in Section \ref{section:background} while the problem at hand is formulated in Section \ref{section:problem}. The observer design and the corresponding stability result are provided in Section \ref{section:main} (main result). Simulation results using an accelerated trajectory are given in Section \ref{section:simulation} to show the performance of the proposed observer and to compare it with a practical adhoc estimation scheme. We conclude the paper with some remarks in Section \ref{section:conclusion}.

\section{Background}\label{section:background}
We denote by $\mathbb{R}$ the set of reals and by $\mathbb{N}$ the set of natural numbers. We denote by $\mathbb{R}^n$ the $n$-dimensional Euclidean space, by $\mathbb{S}^n$ the unit $n$-sphere embedded in $\mathbb{R}^{n+1}$ and by $\mathbb{B}_\epsilon=\{x\in\mathbb{R}^3: \|x\|\leq \epsilon\}$ the closed ball in $\mathbb{R}^3$ with radius $\epsilon$. We use $\|x\|$ to denote the Euclidean norm of a vector $x\in\mathbb{R}^n$ and $\|A\|_F$ to denote the Frobenius norm of a matrix $A\in\mathbb{R}^{n\times n}$. Let $I_n$ be the $n$-by-$n$ identity matrix and let $e_i$ denote the $i-$th column of $I_n$. The Special Orthogonal group of order three is denoted by $\mathbb{SO}(3) := \{ A \in \mathbb{R}^{3\times 3}:\;\mathbf{det}(A)=1,\; AA^{\top}=A^\top A= I_3 \}$. The set $\mathfrak{so}(3):=\left\{\Omega\in\mathbb{R}^{3\times 3}\mid\;\Omega^{\top}=-\Omega\right\}$ denotes the Lie algebra of $\mathbb{SO}(3)$. For $x,~y\in\mathbb{R}^3$, the map $[\cdot]_\times: \mathbb{R}^3\to\mathfrak{so}(3)$ is defined such that $[x]_\times y=x\times y$ where $\times$ is the vector cross-product on $\mathbb{R}^3$. The inverse isomorphism of the map $[\cdot]_\times$ is defined by $\mathrm{vex}:\mathfrak{so}(3)\to\mathbb{R}^3$, such that $\mathrm{vex}([\omega]_\times)=\omega$, for all $\omega\in\mathbb{R}^3$ and $[\mathrm{vex}(\Omega)]_\times=\Omega,$ for all $\Omega\in\mathfrak{so}(3)$. The composition map $\psi := \mathrm{vex}\circ \mathbf{P}_{\mathfrak{so}(3)}$ extends the definition of
        $\mathrm{vex}$ to  $\mathbb{R}^{3\times 3}$, where $\mathbf{P}_{\mathfrak{so}(3)}:\mathbb{R}^{3\times 3}\to\mathfrak{so}(3)$ is the projection map on the Lie algebra $\mathfrak{so}(3)$ such that $\mathbf{P}_{\mathfrak{so}(3)}(A):=(A-A^{\top})/2$. Accordingly, for a given $3$-by-$3$ matrix $A=[a_{ij}]_{i,j = 1,2,3}$, one has ${\textstyle\psi(A)=\frac{1}{2}[a_{32}-a_{23},a_{13}-a_{31},a_{21}-a_{12}]
}$. We define ${\textstyle|R|:=\frac{1}{4}\tr(I_3-R)=\frac{1}{8}\|I_3-R\|_F^2\in[0, 1]}$ as the normalized Euclidean distance on $\mathbb{SO}(3)$. Given a scalar $c>0$, we define the saturation function $\mathbf{sat}_c:\mathbb{R}^n\to\mathbb{R}^n$ such that
$
\mathbf{sat}_c(x):=\min(1,c/\|x\|)x.
$
Given two scalars $c,\epsilon>0$, we also define the smooth projection function $\mathbf{P}_c^\epsilon:\mathbb{R}^3\times\mathbb{R}^3\to\mathbb{R}^3$, found for instance in \cite{Krstic1995NonlinearDesign}, as follows:
\begin{equation}\label{eq:pprojection}
\mathbf{P}_c^\epsilon(\hat\phi,\mu):=
\begin{cases}
\mu,&\textrm{if}\;\|\hat\phi\|<c\;\textrm{or}\;\hat\phi^\top\mu\leq 0,\\
{\scriptstyle
\left(I-\theta(\hat\phi)\frac{\hat\phi\hat\phi^\top}{\|\hat\phi\|^2}\right)}\mu,&\textrm{otherwise},
\end{cases}
\end{equation}
where $\theta(\hat\phi):=\min(1,(\|\hat\phi\|-c)/\epsilon)$.  The projection operator $\mathbf{P}_c^\epsilon(\hat\phi,\mu)$ is locally Lipschitz in its arguments. Moreover, provided that $\|\phi\|\leq c$, the projection map $\mathbf{P}_c^\epsilon(\hat\phi,\mu)$ satisfies, along the trajectories of $\dot{\hat{\phi}}=\mathbf{P}_c^\epsilon(\hat\phi,\mu), \|\hat\phi(t)\|\leq c+\epsilon$, $\forall t\geq 0$. 
\section{Problem Formulation}\label{section:problem}
Consider the following dynamics of a rigid-body vehicle:
\begin{align}
\label{eq:dR}
\dot R&=R[\omega]_\times,\\
\label{eq:dp}
\dot p&=v,\\
\label{eq:dv}
\dot v&=ge_3+Ra_B,
\end{align}
where $p\in\mathbb{R}^3$ is the inertial position of the vehicle's center of gravity, $v\in\mathbb{R}^3$ represents the inertial linear velocity, $R\in \mathbb{SO}(3)$ is the rotation matrix describing the orientation of the body-attached frame with respect to the inertial frame, $\omega$ is the angular velocity of the body-attached frame with respect to the inertial frame expressed in the body-attached frame, $g$ is the norm of the acceleration due to gravity, $e_3 = [0, 0, 1]^\top$ and $a_B=R^\top a_I$ is the ``apparent acceleration", capturing all non-gravitational forces applied to the vehicle, expressed in the body-attached frame. 

We assume available an inertial measurement unit (IMU) that provides measurements in the body frame of the angular velocity, the apparent acceleration and the earth's magnetic field. These sensors are modelled as  follows:
\begin{align}
\omega^y&=\omega+b_\omega,\\
a_B&=R^\top a_I,\\
\label{eq:bm}
m_B&=R^\top m_I,
\end{align}
where $b_\omega$ is a constant unknown gyro bias, $m_I$ is the constant and know earth's magnetic field and $a_I(t)$ is a time-varying unknown apparent acceleration. For the rotational dynamics, the following is a general observability assumption commonly used in attitude estimation.
\begin{assumption}\label{assumption::obsv}
There exists a constant $c_0>0$ such that $\|m_I\times a_I(t)\|\geq c_0$ for all $t\geq 0$.
\end{assumption}
Assumption \ref{assumption::obsv} is guaranteed if the time-varying apparent acceleration $a_I(t)$ is non-vanishing and is \textit{always} non-collinear to the constant magnetic field vector $m_I$. Note that $a_I(t)=0$ corresponds to the rigid body in a free-fall ($\dot v=ge_3$) which is not likely under normal flight conditions. 

We also assume that we have measurements of the following position output vector:
\begin{align}\label{eq:yp}
y=C_pp,
\end{align}
where $C_p\in\mathbb{R}^{m\times 3}, m\in\mathbb{N},$ is a given  output matrix that satisfies 
$
    \mathrm{rank}(C_p)=3.
$
Roughly speaking, this rank assumption means that the measurement $y$ is sufficient for the construction of a converging translational observer (assuming perfect knowledge of the attitude) for \eqref{eq:dp}-\eqref{eq:dv}. The measurement $y$ can be obtained from different possible sensors, depending on the application at hand, that provide some information about the position. For full position measurements obtained from a GPS for instance, $C_p=I_3$. For range measurements, obtained from Ultra-Wide Brand (UWB) sensors for instance, the output $y$ and the matrix $C_p$ can be obtained from the range measurements as done in \cite[Section 4.3.2]{berkane2021nonlinear}; see also the simulation work in Section \ref{section:simulation}.

The objective of this work is to design a full navigation system that processes the measurements \eqref{eq:bm}-\eqref{eq:yp} and outputs reliable estimates for the position $p\in\mathbb{R}^3$, the linear velocity $v\in\mathbb{R}^3$, the orientation $R\in\mathbb{SO}(3)$, and the gyro bias $b_\omega\in\mathbb{R}^3$. More specifically, our goal is to design an exponentially convergent nonlinear observer that estimates the whole state of the vehicle using an observer state that evolves on the same manifold as the system state ($\mathbb{SO}(3)\times\mathbb{R}^9$). We further consider the following mild (realistic) constraints on the trajectory of the vehicle  which are needed to prove the main result:
\begin{assumption}\label{assumption::bounded_ra}
There exist constants $c_1,c_2,c_3>0$ such that $c_1\leq\|a_I(t)\|\leq c_2$ and $\|\dot a_I(t)\|\leq c_3$ for all $t\geq 0$.
\end{assumption}
\begin{assumption}\label{assumption::bounded_bw}
There exists constants $c_4,c_5>0$ such that $\|\omega(t)\|\leq c_4$ and $\|b_\omega\|\leq c_5$ for all $t\geq 0$.
\end{assumption}
\section{Observer Design and Stability Analysis}\label{section:main}
In this section, we design a navigation observer to estimate the full state $(p,v,R,b_\omega)$ consisting of the position, velocity, orientation and gyro bias. We define the combined translational state $x:=[p^\top,v^\top]^\top\in\mathbb{R}^6$. Then, in view of \eqref{eq:dp}-\eqref{eq:dv} and \eqref{eq:yp}, the dynamics of $x$ are written as: 
\begin{align}
\label{eq:dx}
\dot x&=Ax+B(ge_3+Ra_B),\\
\label{eq:y}
y&=Cx,
\end{align}
where the matrices $A, B$ and $C$ are defined  as follows:
\begin{align}
A=\begin{bmatrix}
0_{3\times 3}&I_3\\
0_{3\times 3}&0_{3\times 3}
\end{bmatrix},
B:=\begin{bmatrix}
0_{3\times 3}\\
I_3
\end{bmatrix}, 
C=\begin{bmatrix}
C_p^\top\\
0_{3\times m}
\end{bmatrix}^\top.
\end{align}
The translational dynamics \eqref{eq:dx}-\eqref{eq:y} are those of a linear time-invariant system with unknown input $a_I=Ra_B$  representing the apparent acceleration which is known only in the body-frame. A common approach, in practice, consists in  assuming that the acceleration of the vehicle is negligible, {\it i.e.,} $\dot v\approx 0$, and therefore $a_I\approx -ge_3$. Under this small linear acceleration assumption, the attitude and gyro bias estimation can be done separately using the following explicit complementary filter proposed in \cite{Mahony2008NonlinearGroup,Grip2013}:
\begin{align}
\label{eq:dhatR-adhoc}
\dot{\hat R}&=\hat R[\omega_y-\hat b_\omega+k_R\sigma_R]_\times,\\
\dot{\hat b}_\omega&=\mathbf{P}_{c_5}^{\epsilon_b}(\hat b_\omega,-k_b\sigma_R),
\end{align}
with the innovation term
\begin{align}
\sigma_R&=\rho_1(m_B\times\hat R^\top m_I)+\rho_2(a_B\times\hat R^\top(-ge_3)).
\end{align}
Once the attitude is estimated, the translational motion state $x$ can be estimated using a Luenberger-like observer as follows: 
\begin{align}
\label{eq:dhatx-adhoc}
    \dot{\hat x}&=A\hat x+B(ge_3+\hat Ra_B)+K(y-C\hat x),
\end{align}
where $K$ is a gain matrix guaranteeing that $(A-KC)$ is Hurwitz, which can  be either constant or tuned via a Riccati equation such as in  the Kalman filter. The cascaded adhoc estimation scheme \eqref{eq:dhatR-adhoc}-\eqref{eq:dhatx-adhoc} will be used as a state-of-the-art algorithm in our comparison results in Section \ref{section:simulation}.    

In this work, however, we design our estimation algorithm without the small acceleration assumption which compromises the performance of the adhoc scheme when the vehicle is subject to large linear accelerations. In particular, we propose the following nonlinear navigation observer on $\mathbb{SO}(3)\times\mathbb{R}^9$:
\begin{align}
\label{eq:dRhat}
\dot{\hat R}&=\hat R[\omega_y-\hat b_\omega+k_R\sigma_R]_\times,\\
\label{eq:dbhat1}
\dot{\hat b}_\omega&=\mathbf{P}_{c_5}^{\epsilon_b}(\hat b_\omega,-k_b\sigma_R),\\
\label{eq:dxhat1}
\dot{\hat x}&=A\hat x+B(ge_3+\hat Ra_B)+K(y-C\hat x)+\sigma_x,
\end{align}
with initial conditions $\hat x(0)\in\mathbb{R}^6, \hat R(0)\in\mathbb{SO}(3)$ and $\hat b_\omega(0)\in\mathbb{B}(c_5+\epsilon_b)$. The innovation terms  $\sigma_x:=[\sigma_p^\top,\sigma_v^\top]^\top\in\mathbb{R}^6$  and $\sigma_R\in\mathbb{R}^3$ are defined as follows:
\begin{align}
\label{eq:sigmav}
\sigma_v&=K_pC_p\sigma_p,\\
\label{eq:sigmap}
\sigma_p&=k_R(K_vC_p)^{-1}[\hat R\sigma_R]_\times\hat Ra_B,\\
\label{eq:sigma2:1}
\sigma_R&=\rho_1(m_B\times\hat R^\top m_I)+\rho_2(a_B\times\hat R^\top\mathbf{sat}_{\hat c_2}(K_v(y-C\hat x))).
\end{align}
The constant scalars $k_R,k_b,\rho_1,\rho_2,\epsilon_b,\hat c_2$ are positive with $\hat c_2>\sqrt{8}c_2$, the gain matrix $K=[K_p^\top\; K_v^\top]^\top\in\mathbb{R}^{6\times m}$ is chosen as $K=L_\gamma K_0$ such that $A-K_0C$ is Hurwitz, $L_\gamma=\textrm{blockdiag}(\gamma I_3,\gamma^2 I_3)$ and $\gamma\geq 1$. Note that the inverse of the matrix $(K_vC_p)$ in \eqref{eq:sigmap} exists since $(A-K_0C)$ is a $2\times 2$ invertible block matrix and the Schur complement of $I_3$ is the matrix $-K_vC_p$; see \cite[Theorem 2.2]{lu2002inverses}. \begin{figure}
    \centering
    \includegraphics[width=\columnwidth]{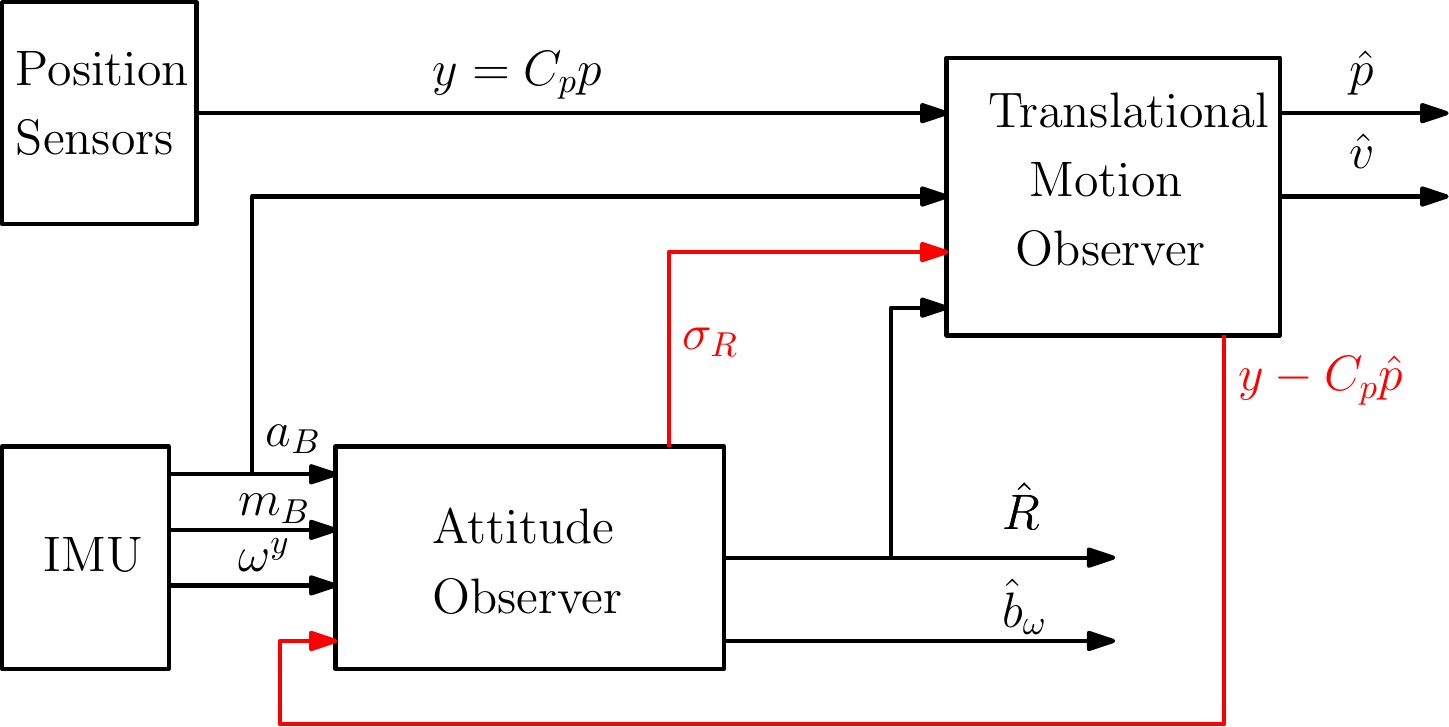}
    \caption{Overall structure of the proposed nonlinear observer in \eqref{eq:dRhat}-\eqref{eq:dxhat1}. The observer uses an IMU and position information to estimate the whole state of the vehicle. Compared to the adhoc estimation scheme \eqref{eq:dhatR-adhoc}-\eqref{eq:dhatx-adhoc}, the proposed observer introduces coupling (in bold red) between the translational estimator and the rotational estimator through their innovation terms. This additional coupling is important to guarantee the stability of the observer without the low acceleration assumption.}
    \label{fig:structure}
\end{figure}
The overall structure of the proposed estimation algorithm is depicted in Fig.~\ref{fig:structure}. Compared to the adhoc estimation scheme (12)-(15), the proposed observer introduces coupling through $\sigma_x$ and $\sigma_R$ between the translational estimator and  the  rotational  estimator.  This  additional coupling  is  important  to  guarantee  the  stability  of  the  observer  without  the low acceleration assumption. We define the following estimation errors:
\begin{align}
\label{eq:tildex}
\tilde x&:=x-\hat x,\\
\tilde R&:=R\hat R^\top,\\
\label{eq:tildeb}
\tilde b_\omega&:=b_\omega-\hat b_\omega.
\end{align}
In view of  \eqref{eq:dR}-\eqref{eq:dp}\, \eqref{eq:dbhat1}-\eqref{eq:dxhat1} and \eqref{eq:tildex}-\eqref{eq:tildeb} one can derive the following dynamics for the estimation errors:
\begin{align}
\label{eq:dtildex}
\dot{\tilde x}&=(A-KC)\tilde x+B(I-\tilde R)^\top a_I-\sigma_x,\\
\label{eq:dtildeR}
\dot{\tilde R}&=\tilde R[-\hat R(\tilde b_\omega+k_R\sigma_R)]_\times,\\
\label{eq:dtildebw}
\dot{\tilde b}_\omega&=\mathbf{P}_{c_5}^{\epsilon_b}(\hat b_\omega,k_b\sigma_R).
\end{align}
To prove the convergence of the estimation errors, we introduce the following auxiliary error variable
\begin{align}\label{eq:zeta}
\zeta:=L_\gamma^{-1}\left[(A-KC)\tilde x+B(I-\tilde R)^\top a_I\right]. 
\end{align}
The term between brackets in \eqref{eq:zeta} is an estimation error term that contains a suitable coupling between the translational estimation error $\tilde x$ and the rotational estimation error $\tilde R$. This coupling in the estimation errors, obtained from \eqref{eq:dtildex},  allows to design the innovation term $\sigma_x$ by inspecting the dynamics of this coupled estimation error. This coupling is instrumental in showing (exponential) stability of the overall closed-loop system without requiring an additional observer state as done in our previous work \cite{berkane2021nonlinear}.  Moreover, the matrix $L_\gamma^{-1}$ is introduced to assign a certain time-scaling structure between the different estimation errors in the same spirit as in \cite{berkane2021nonlinear}; see also \cite{esfandiari1992output,saberi1990} for a motivation of this use in the context of high-gain observers. The high-gain $\gamma$ allows to deal efficiently with the interconnection between the rotational and translational dynamics. Now we are ready to state our main result. \begin{theorem}\label{theorem:navigation}
Consider the interconnection of the dynamics \eqref{eq:dp}-\eqref{eq:dR} with the observer \eqref{eq:dxhat1}-\eqref{eq:sigma2:1} where Assumptions \eqref{assumption::obsv}-\eqref{assumption::bounded_bw} are satisfied. For each $\epsilon\in(\frac{1}{2},1)$ and for all initial conditions such that $\zeta(0)\in\mathbb{R}^6$, $\tilde R(0)\in\mho(\epsilon)=\{\tilde R: |\tilde R(0)|\leq\epsilon\}$ and $\hat b_\omega(0)\in\mathbb{B}(c_5+\epsilon_b)$, there exist $k_R^*>0$ and $\gamma^*\geq 1$ such that, for all $k_R\geq k_R^*$ and $\gamma\geq\gamma^*$, the estimation errors $(\zeta,\tilde R, \tilde b_\omega)$ are globally uniformly bounded and converge exponentially to zero.
\end{theorem}
\begin{proof}
See Appendix \ref{proof:theorem:navigation}.
\end{proof}
Theorem \ref{theorem:navigation} shows that the proposed navigation observer \eqref{eq:dRhat}-\eqref{eq:sigma2:1} guarantees exponential convergence of the estimation errors to zero starting from any initial condition inside the set $\mho(\epsilon)$. As $\epsilon$ tends to $1$, the set $\mho(\epsilon)$ covers the set of all attitude errors with rotation angle less than $180^\circ$. The derived bounds on the gains $k_R$ and $\gamma$ (see the proof) are expected to be very conservative as demonstrated in the simulation results. 

\section{Simulation}\label{section:simulation}
In this section, we simulate the proposed nonlinear estimator presented in Section \ref{section:main} in the case of position range measurements. The real position trajectory of the vehicle is given by:
\begin{align}
p(t)=\begin{bmatrix}
\cos(2\pi t^2/100)\\
\sin(2\pi t^2/100)\\
1
\end{bmatrix}.
\end{align}
This corresponds to a circular trajectory with an angular frequency that increases over time (given by $t/100$). The attitude trajectory is generated using the following angular velocity:
\begin{align}
\omega(t)=\begin{bmatrix}
\sin(0.2t)\\
\cos(0.1t)\\
\sin(0.3t+\pi/6)
\end{bmatrix},
\end{align}
and an initial attitude $R(0)=\exp([\pi e_1/2]_\times)$. The gyro measurements are corrupted by a constant bias of $3\;\mathrm{(deg/sec)}$ in each axis. The inertial earth's magnetic field is taken as $m_I=[0.033\;0.1\;0.49]^\top$ and the earth's gravity is $g=9.81\;(\mathrm{m/sec}^2)$. We also assume available four non-coplanar source points for range sensing located at:
\begin{align}
    a_1&=[1\;1\;2]^\top,\\
    a_2&=[1 \;3 \;\;0]^\top,\\
    a_3&=[0 \;1 \;1]^\top,\\
    a_4&=[6 \;5 \;5]^\top.
\end{align}
The corresponding range measurements are given by
\begin{align}
d_i=\|p-a_i\|,\quad i=1,\cdots,4.
\end{align}
To obtain an output equation of type \eqref{eq:yp} we proceed as follows. Define the following measurable scalars
\begin{align}
y_i=\frac{1}{2}\left(d_i^2-d_1^2-\|a_i\|^2+\|a_1\|^2\right),\quad i=2,\cdots,4.
\end{align}
Then, one can show that $y_i=(a_1-a_i)^\top p, i=2,\cdots,4$. Define $y=[y_2,\cdots,y_4]^\top$ one has
\begin{align}\label{eq:Cp:range}
y=\begin{bmatrix}
(a_1-a_2)^\top\\
(a_1-a_3)^\top\\
(a_1-a_4)^\top
\end{bmatrix}p:=C_pp.
\end{align}
It can be verified that $C_p$ has a rank of $3$.

 The initial conditions for the observer states are $\hat p(0)=\hat v(0)=\hat b_\omega(0)=0$ and $\hat R(0)=I_3$. The parameters of the observer are selected as $k_R=2, k_b=1, \rho_1=\rho_2=1,\epsilon_b=0.001, \hat c_2=9\sqrt{8}$ and $\gamma=2$. The gain $K_0$ is tuned such that the matrix $(A-K_0C)$ has eigenvalues at $\{-3,-3,-3,-4,-4,-4\}$.  The adhoc estimator \eqref{eq:dhatR-adhoc}-\eqref{eq:dhatx-adhoc} is implemented with the same parameters and initial conditions. The simulation results given in Figures \ref{fig:trajectory}-\ref{fig:bias} show that the proposed observer is able to estimate the position, velocity, acceleration, attitude and gyro bias using IMU and range measurements. On the other hand, as the acceleration of the vehicle increases, the adhoc estimator drifts away from the true trajectory while the proposed estimator is stable against high accelerations. 
 \begin{figure}
    \centering
    \includegraphics[width=\columnwidth]{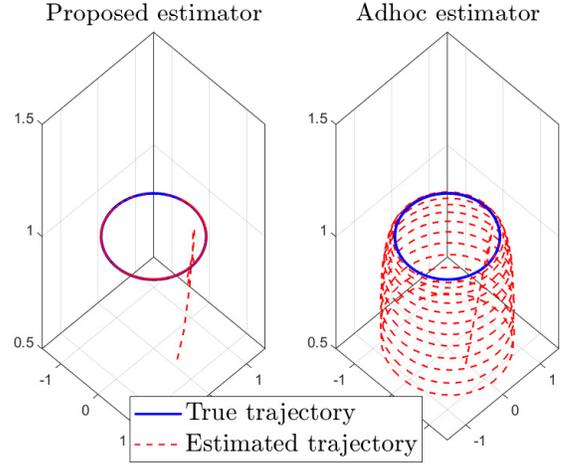}
    \caption{True and estimated trajectory of the vehicle for the proposed observer which is compared against the adhoc estimator that assumes negligible accelerations. The video of the simulation can be found at \url{https://youtu.be/zbkSDZgh3vU}.}
    \label{fig:trajectory}
\end{figure}

 \begin{figure}
    \centering
    \includegraphics[width=\columnwidth]{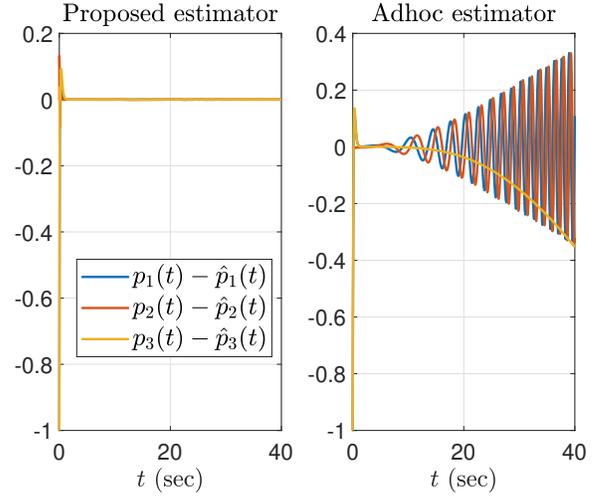}
    \caption{Position estimation error.}
    \label{fig:position}
\end{figure}

\begin{figure}
    \centering
    \includegraphics[width=\columnwidth]{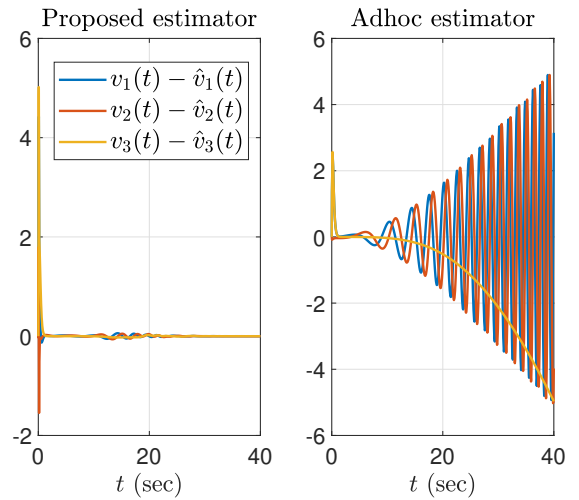}
    \caption{Velocity estimation error.}
    \label{fig:velocity}
\end{figure}

\begin{figure}
    \centering
    \includegraphics[width=\columnwidth]{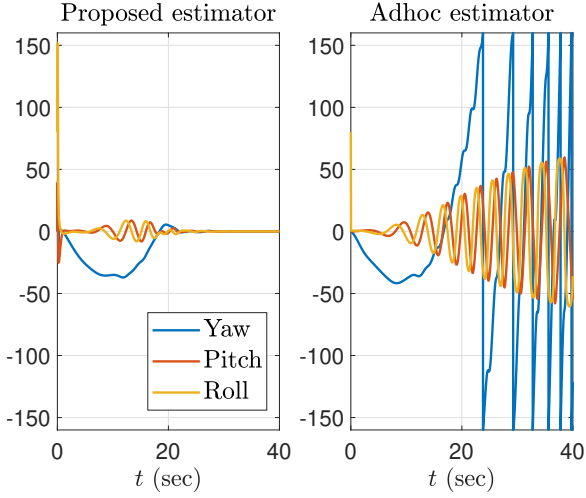}
    \caption{Euler angles (Roll, Pitch, Yaw) corresponding to the attitude estimation error $\tilde R=R\hat R^\top$.}
    \label{fig:rotation}
\end{figure}

\begin{figure}
    \centering
    \includegraphics[width=\columnwidth]{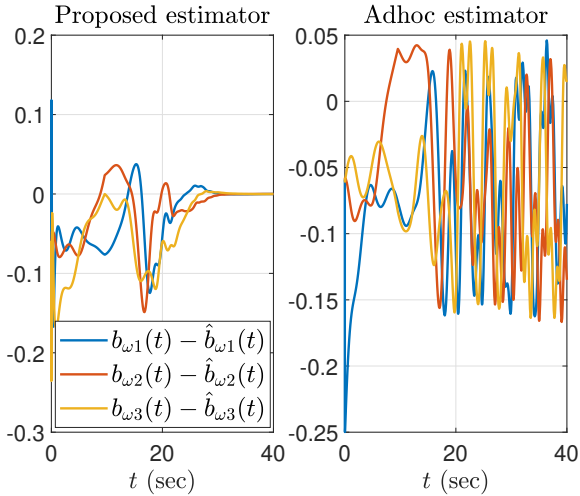}
    \caption{Gyro bias estimation error.}
    \label{fig:bias}
\end{figure}

\section{Conclusion}\label{section:conclusion}
In this paper we solved the full state (position, velocity, orientation, gyro bias) estimation problem for vehicles navigating in $3$-dimensional spaces. The proposed solution is based on a nonlinear observer evolving on the configuration space $\mathbb{SO}(3)\times\mathbb{R}^9$ that uses IMU and position measurements as inputs and guarantees semi-global exponential stability. The proposed solution is particularly suitable in applications with important accelerations where the traditional cascaded approach (attitude estimation + position estimation) fails due to the small-accelerations assumption. The key feature of this observer is the fact that it does not add computational complexity compared to the traditional approach whereas our previous solution in \cite{berkane2021nonlinear} requires the introduction of auxiliary states to achieve the same stability result.
\appendices
\section{Proof of Theorem \ref{theorem:navigation}}\label{proof:theorem:navigation}
The proof of this result is mainly inspired from our proof in \cite{berkane2021nonlinear} with few modifications related to the particular closed-loop dynamics at hand. First, we provide the following easy-to-check facts: 
\begin{align}
L_\gamma^{-1} AL_\gamma=\gamma A,\quad L_\gamma^{-1} B&=\gamma^{-2}B,\quad CL_\gamma=\gamma C.
\end{align}
Moreover, in view of the expression of $\sigma_x$ in \eqref{eq:sigmav}-\eqref{eq:sigmap}, it is not difficult to show that 
\begin{align}
    (A-KC)\sigma_x=-k_RB[\hat R\sigma_R]_\times\hat Ra_B.
\end{align}
Now using the above facts, and in view of \eqref{eq:dtildex}-\eqref{eq:dtildebw}, the time-scaled dynamics of $\zeta$ are given by
\begin{align}
\frac{1}{\gamma}\dot{\zeta}
&=(A-K_0C)\zeta+\frac{1}{\gamma^3}Bg(t,\tilde R,\tilde b_\omega),
\label{eq:dzeta}
\end{align}
where we have defined $g(t,\tilde R,\tilde b_\omega):=(I-\tilde R)^\top\dot{a}_I+\tilde R^\top[a_I(t)]_\times R(t) \tilde b_\omega$. The term $g(t,\tilde R,\tilde b_\omega)$ is {\it a priori} bounded by some constant $c_g>0$ in view of Assumptions \ref{assumption::bounded_ra}-\ref{assumption::bounded_bw} and the property of the projection mechanism that guarantees boundedness of the bias estimation error $\tilde b_\omega$. Consider the following real-valued positive function:
\begin{align}\label{eq:navigation:Vzeta}
\mathbf{V}(\zeta)=\frac{1}{\gamma}\zeta^\top P\zeta,
\end{align}
where $P$ is the unique solution for the Lyapunov equation $P(A-K_0C)+(A-K_0C)^\top P=-I$ which is positive definite thanks to the fact that $(A-K_0C)$ is Hurwitz. Denote by $\beta_1$ and $\beta_2$ the smallest and largest eigenvalues of $P$, respectively. The time derivative of $\mathbf{V}(\zeta)$ along the trajectories of \eqref{eq:dzeta} satisfies
\begin{align}
\dot{\mathbf{V}}(\zeta)
&=-\|\zeta\|^2+\frac{2}{\gamma^3}\zeta^\top PBg(t,\tilde R,\tilde b_\omega),\\
&\leq-\frac{1}{2}\|\zeta\|^2+\|\zeta\|\left(	\frac{2\beta_2c_g}{\gamma^3}-\frac{1}{2}\|\zeta\|\right),\\
&\leq-\frac{1}{2}\|\zeta\|^2,\quad\forall\|\zeta\|\geq\frac{4\beta_2c_g}{\gamma^3},\\
&\leq-\frac{\gamma}{2\beta_2}\mathbf{V}(\zeta),\quad\forall\|\zeta\|\geq\frac{4\beta_2c_g}{\gamma^3}.
\label{eq:dV_zeta}
\end{align}
Consider the set $\Omega(c_\zeta):=\{\zeta: \mathbf{V}(\zeta)\leq \gamma^{-5}c_\zeta^2\beta_1\}$ for some $c_\zeta>0$. Let us pick $\gamma\geq 4\beta_1^{-\frac{1}{2}}\beta_2^{\frac{3}{2}}c_gc_\zeta^{-1}$. Then, it can be shown that when $\zeta$ is outside the set $\Omega(c_\zeta)$ one has $\|\zeta\|> 4\beta_2c_g\gamma^{-3}$ and therefore, by \eqref{eq:dV_zeta}, the function $V(\zeta)$ is exponentially decreasing outside $\Omega(c_\zeta)$. In this case, $\zeta$ must enter the set $\Omega(c_\zeta)$ before the following time:
\begin{align}\label{eq:T*}
\bar T^*=\frac{2\beta_2}{\gamma}\ln\left(\frac{\gamma\mathbf{V}(\zeta(0))}{\beta_1c_\zeta^2}\right),
\end{align}
which can be tuned arbitrary small by increasing the value of $\gamma$. On the other hand, when $\zeta\in\Omega$ one has
$\|\zeta\|^2\leq\gamma\beta_1^{-1}\mathbf{V}(\zeta)\leq ( c_\zeta\gamma^{-2})^2$. The following result immediately follows:
\begin{multline}
    \label{fact:bound:zeta}
\forall c_\zeta, T>0,\forall \zeta(0), \exists\gamma_1\geq 1\;\textrm{s.th.}\; \gamma\geq\gamma_1\Rightarrow\\ \|\zeta(t)\|\leq\gamma^{-2}c_\zeta,\;\forall t\geq T.
\end{multline}
This result shows that the gain $\gamma$ can be tuned to guarantee that the error variable $\zeta$ converges arbitrary fast to an arbitrary small ball. Now, we show that the gains can be tuned to guarantee forward invariance of the set $\mho(\epsilon)$. Let $\epsilon\in(\frac{1}{2},1)$ and let the initial attitude estimation error be such that $\tilde R(0)\in\mho(\epsilon)$. The time derivative of $|\tilde R|^2=\mathrm{tr}(I-\tilde R)/4$, in view of \eqref{eq:dtildeR} and making use of \cite[Lemma 2]{Berkane2017HybridSO3}, satisfies
\begin{align}
\nonumber
\frac{d}{dt}|\tilde R|^2&=-\frac{1}{4}\mathrm{tr}(\dot{R})\\
&=-\frac{1}{4}\mathrm{tr}(\mathbf{P}_{\mathfrak{so}(3)}(\tilde R)[-\hat R(\tilde b_\omega+k_R\sigma_R)]_\times)\\
\nonumber
&=-\frac{1}{2}\psi(\tilde R)^\top\hat R(\tilde b_\omega+k_R\sigma_R)\\
&\leq\|\tilde b_\omega\|+k_R\|\sigma_R\|\\
&\leq c_b+k_R(\rho_1\|m_I\|^2+\rho_2c_2\hat c_2):=c_R.
\end{align}
In view of the above inequality on the velocity of $|\tilde R|^2$, it can be deduced that the minimum time necessary for the attitude estimation error $\tilde R$ to leave the set $\mho(\epsilon)$ is greater than 
$$
t_R:=\frac{\epsilon^2-|\tilde R(0)|^2}{c_R}.
$$
Since we have knowledge about the minimum time necessary for the attitude estimation error to leave the set $\mho(\epsilon)$, it is possible to prevent such a scenario by imposing some (high gain) conditions on the gains $\gamma$ and $k_R$ as shown next. Pick  $0<T\leq t_R$ and $c_\zeta=\min(\bar c_\zeta/k_R,\hat c_2-\sqrt{8}c_2)$ for some arbitrary $\bar c_\zeta>0$. By \eqref{fact:bound:zeta} there exists $\gamma^*\geq1$ such that if one chooses $\gamma\geq\gamma^*$ then $\|\zeta(t)\|\leq \gamma^{-2}c_\zeta$ for all $t\geq T$. On the other hand, in view of \eqref{eq:zeta}, we have
\begin{align}
    K_v(y-C\hat x)=-B^\top L_\gamma\zeta+(I-\tilde R)^\top a_I.
\end{align}
Therefore, for all $t\geq T$, one has
\begin{align}
\|K_v(y-C\hat x)\|
&\leq \gamma^2\|\zeta\|+\sqrt{8}c_2\leq c_\zeta+\sqrt{8}c_2\leq \hat c_2.
\end{align}
Consequently, for all $t\geq T$, the saturation function $\mathbf{sat}_{\hat c_2}(\cdot)$ in \eqref{eq:sigma2:1} can be removed. Therefore, using \cite[Proposition 3]{Berkane2017ConstructionStabilization}, the innovation term $\sigma_R$ in \eqref{eq:sigma2:1} can be written as follows:
\begin{align}
\sigma_R
&=\rho_1(m_B\times\hat R^\top m_I)+\rho_2(a_B\times\hat R^\top (a_I-B^\top L_\gamma\zeta))\\
\label{eq:navigation:sigmaR2}
&=2\hat R^\top\psi(M\tilde R)-\rho_2(a_B\times\hat R^\top B^\top L_\gamma\zeta),
\end{align}
with $M:=\rho_1 m_Im_I^\top+\rho_2a_Ia_I^\top$ which is positive semidefinite with an rank of $2$ (by Assumption \ref{assumption::obsv}). Using the above expression of $\sigma_R$, the time derivative of $|\tilde R|^2$ satisfies
\begin{align}
\frac{d}{dt}|\tilde R|^2
&=-k_R\psi(\tilde R)^\top\psi(M\tilde R)-\frac{1}{2}\psi(\tilde R)^\top(\hat R\tilde b_\omega-\\
&\qquad k_R\rho_2(\hat Ra_B)\times B^\top L_\gamma\zeta)\\
&\leq-4k_R\lambda_{\min}^{\E(M)}|\tilde R|^2(1-|\tilde R|^2)+|\tilde R|\|\tilde b_\omega\|+\\
&\qquad\gamma k_R\rho_2c_2|\tilde R|\|\zeta\|\\
&\leq -4k_R\lambda_{\min}^{\E(M)}|\tilde R(t)|^2(1-|\tilde R(t)|^2)+c_b+\rho_2c_2\bar c_\zeta
\end{align}
for all $t\geq T$, where inequalities from \cite[Lemma 2]{Berkane2017HybridSO3} have been used with $\E(M):=\frac{1}{2}(\tr(M)-M^\top)$ for any $M$. Note that the matrix $\E(M)$ is positive definite in view of Assumption \ref{assumption::obsv}.  Now assume that $|\tilde R(t)|=\epsilon$ and $k_R>(c_b+\rho_2c_2\bar c_\zeta)/(4\lambda_{\min}^{\E(M)}\epsilon^2(1-\epsilon^2))$ then one has
\begin{align*}
\frac{d}{dt}|\tilde R(t)|^2&\leq-4k_R\lambda_{\min}^{\bar A}\epsilon^2(1-\epsilon^2)+c_b+\rho_2c_2\bar c_\zeta<0,
\end{align*}
for all $t\geq T$. This implies that $|\tilde R(t)|$ is strictly decreasing whenever $|\tilde R(t)|=\epsilon$. It follows from the continuity of the solution that $\tilde R(t)$ will never leave the ball $\mho(\epsilon)$ for all $t\geq T$. Recall also that $|\tilde R(t)|\leq\epsilon$ for all $t\leq T$ (since $T\leq t_R$). This implies that the set $\mho(\epsilon)$ is forward invariant. Now, let us show the exponential convergence after $t\geq T$. Consider the following Lyapunov function candidate:
\begin{multline}
\label{V}
\mathbf{W}(\zeta,\tilde R,\hat R,\tilde b_\omega)=|\tilde R|^2+\frac{\mu k_R}{2k_b}\tilde b_\omega^\top\tilde b_\omega+\\\mu\tilde b_\omega^\top\hat R^\top\psi(\tilde R)+\gamma^5\mathbf{V}(\zeta),
\end{multline}
where $\mu$ is some positive constant scalar and $\mathbf{V}(\zeta)$ is defined in \eqref{eq:navigation:Vzeta}. Let $\varsigma:=[|\tilde R|,\|\tilde b_\omega\|,\|\zeta\|]^\top$ and following a similar procedure to \cite[(A.11)-(A.14)]{berkane2021nonlinear}, we can show that there exist positive definite matrices $P_1,P_2,P_3$ such that the $\varsigma^\top P_1\varsigma\leq \mathbf{W}\leq \varsigma^\top P_2 \varsigma$ and $\dot{\mathbf{W}}\leq-\varsigma^\top P_3 \varsigma$ provided that the gains satisfy  
\begin{align*}
\mu&<\lambda_{\min}^{\E(M)}(1-\epsilon^2)/\alpha_2,\\
k_R&>\max\left\{2\mu k_b,\frac{2\alpha_1\mu^2+(1+2c_\omega\mu)^2}{2\mu\lambda_{\min}^{\E(M)}(1-\epsilon^2)}\right\},\\
\gamma&>\max\left\{\frac{4(\beta_2)^2c_2^2}{\mu}, \frac{k_R\rho_2c_2+4\sqrt{2}\beta_2c_3+\mu(\alpha_3+k_R\alpha_4)^2}{4k_R\lambda_{\min}^{\E(M)}(1-\epsilon^2)}\right\},
\end{align*}
where $\alpha_1=8c_b^2+4k_b\lambda_{\max}^{\E(M)}$, $\alpha_2=8\lambda_{\max}^{\E(M)}c_b(\sqrt{2}+4))$, $\alpha_3=2k_b\rho_2c_2$ and $\alpha_4=2\rho_2c_bc_2(\sqrt{2}+4)$. Consequently, the error variable $\varsigma$ converges exponentially to zero after $t\geq T$.
\bibliographystyle{IEEEtran}
\bibliography{ref.bib}
 \end{document}